\documentclass[aps,prx,reprint,longbibliography,amsmath,amssymb,superscriptaddress,showkeys]{revtex4-1}
\usepackage[breaklinks,colorlinks=true]{hyperref}
\usepackage[T1]{fontenc}
\usepackage[utf8]{inputenc}
\usepackage{mathtools}
\usepackage{amsmath,amsthm,amssymb}
\usepackage[dvipsnames]{xcolor}
\usepackage{bm}
\usepackage{dsfont}
\usepackage{lipsum}
\usepackage{calrsfs}
\usepackage{tikz-cd}
\usepackage{algorithm}
\usepackage{minted}
\usepackage{xspace}
\allowdisplaybreaks
\usepackage{tikz}
\usetikzlibrary{external}

\def\compiletikz{0}

\if1\compiletikz
\fi

\newcommand\myshade{85}
\colorlet{mylinkcolor}{BrickRed}
\colorlet{mycitecolor}{NavyBlue}
\colorlet{myurlcolor}{Aquamarine}

\hypersetup{
  linkcolor  = mylinkcolor!\myshade!black,
  citecolor  = mycitecolor!\myshade!black,
  urlcolor   = myurlcolor!\myshade!black,
  colorlinks = true,
}

\DeclareMathAlphabet{\pazocal}{OMS}{zplm}{m}{n}

\newcommand\mydots{\hbox to 1em{.\hss.\hss.}}

\newcommand{\gdot}[0]{\!\cdot\!}

\def\equationautorefname#1#2\null{Eq.#1(#2\null)}

\DeclareMathOperator{\diag}{diag}
\makeatletter
\DeclareRobustCommand{\cev}[1]{%
  \mathpalette\do@cev{#1}%
}
\newcommand{\do@cev}[2]{%
  \fix@cev{#1}{+}%
  \reflectbox{$\m@th#1\vec{\reflectbox{$\fix@cev{#1}{-}\m@th#1#2\fix@cev{#1}{+}$}}$}%
  \fix@cev{#1}{-}%
}
\newcommand{\fix@cev}[2]{%
  \ifx#1\displaystyle
    \mkern#23mu
  \else
    \ifx#1\textstyle
      \mkern#23mu
    \else
      \ifx#1\scriptstyle
        \mkern#22mu
      \else
        \mkern#22mu
      \fi
    \fi
  \fi
}
\makeatother

\newcommand{\past}[2]{\cev{\bm #1}_{#2}}

\newcommand{\finfut}[3]{\vec{\bm #1}^{#3}_{#2}}

\begin{document}

\title{Symmetries at the origin of hierarchical emergence}

\author{Fernando E. Rosas}
\email{f.rosas@sussex.ac.uk}
\affiliation{Sussex~AI~and~Sussex~Centre~for~Consciousness~Science,~Department~of~Informatics,~University~of~Sussex}
\affiliation{Centre for Complexity Science and Department of Brain Sciences, Imperial College London}
\affiliation{Centre for Eudaimonia and Human Flourishing, University of Oxford}

\newtheorem{definition}{Definition}
\newtheorem{conjecture}{Conjecture}
\newtheorem{theorem}{Theorem}
\newtheorem{lemma}{Lemma}
\newtheorem{proposition}{Proposition}
\newtheorem{corollary}{Corollary}
\newtheorem{example}{Example}
\newtheorem{remark}{Remark}

\begin{abstract}

\noindent
Many systems of interest exhibit nested emergent layers with their own rules and regularities, and our knowledge about them seems naturally organised around these levels. 
This paper proposes that this type of hierarchical emergence arises as a result of underlying symmetries. 
By combining principles from information theory, group theory, and statistical mechanics, one finds that dynamical processes that are equivariant with respect to a symmetry group give rise to emergent macroscopic levels organised into a hierarchy determined by the subgroups of the symmetry. 
The same symmetries happen to also shape Bayesian beliefs, yielding hierarchies of abstract belief states that can be updated autonomously at different levels of resolution. 
These results are illustrated in Hopfield networks and Ehrenfest diffusion, showing that familiar macroscopic quantities emerge naturally from their symmetries. 
Together, these results suggest that symmetries provide a fundamental mechanism for emergence and support a structural correspondence between objective and epistemic processes, making feasible inferential problems that would otherwise be computationally intractable.
\end{abstract}

\maketitle

\section{Introduction}

Our ability to successfully navigate and learn about the world, despite the limitations imposed by `no free lunch' theorems~\cite{wolpert1996lack,adam2019no}, suggests that our world has special properties that make this success possible. 
A key factor enabling this seems to be that the world displays nested yet distinct layers of organisation --- a property often called \emph{emergence}~\cite{jensen2022complexity,clayton2006re,carroll2024emergence}. 
Emergence gives rise to a separation of scales that allows each level to exhibit its own rules and regularities, which can be learned while disregarding the activity of levels below~\cite{rosas2024software,krakauer2025large}. 
This hierarchical architecture is mirrored in the organisation of knowledge into scientific disciplines (physics $\to$ chemistry $\to$ biology $\to$ sociology and economics), where each field abstracts away details of the levels below to form its own effective theories~\cite{anderson1972more}.

Emergence has been studied in statistical physics via phase transitions and renormalisation group theory, which provide a powerful framework that describes how separation of scales occurs and what the resulting emergent laws are (i.e., universality)~\cite{kadanoff1966scaling,wilson1979problems,sole2011phase,dupuis2021nonperturbative}. 
Complementing this work, a line of investigations inspired by theoretical neuroscience has focused on multiscale dynamics, characterising various aspects of emergence in terms of information-theoretic principles~\cite{seth2010measuring,hoel2013quantifying,rosas2020reconciling,barnett2023dynamical}. 
These approaches provide rigorous tools to identify when emergent levels arise in general systems far from the thermodynamic limit, but have not yet identified fundamental principles that can explain how this happens. 
Furthermore, all of these approaches are limited in the degree to which they account for how emergence enhances learnability and inference.

Building on these developments, this paper proposes symmetry --- more precisely, dynamical equivariance --- as a fundamental driver of hierarchical emergence in both objective and epistemic processes. 
The results presented here show that when the dynamics of a stochastic process commute with a symmetry group, they naturally induce coarse-grained variables whose evolution is informationally self-contained, so that microscopic details of their instantiation are irrelevant for predicting their future. 
Moreover, the algebraic structure of the symmetry group reveals a whole hierarchy of emergent variables. Each subgroup corresponds to a coarser or finer partition of the microscopic state space, and the subgroup lattice is mirrored by a partial order over informationally closed macroscopic levels.

Crucially, the same symmetries that generate a hierarchy in objective processes also shape the structure of Bayesian beliefs held by an observer who only has partial access to the system. 
The hierarchical nature of these beliefs allows tracking macroscopic properties of high-dimensional systems while disregarding low-level details, making otherwise unfeasible inference problems tractable. 
These ideas are illustrated on paradigmatic models from statistical mechanics and computational neuroscience, where familiar macroscopic quantities (such as magnetisation or particle count) emerge exactly as informationally closed macroscopic processes naturally arising from underlying symmetries.

\section{Symmetry and hierarchical emergence}

\subsection{Emergence as informational closure} 

Consider a system whose state is specified by the random variable $X_t\in\mathcal{X}$ that is measured at discrete timepoints $t\in\mathbb{Z}$. A coarse-graining mapping $\phi:\mathcal{X}\to\mathcal{Z}$ gives rise to another process $Z_t=\phi(X_t)$. 
This implies that any change in $Z_t$ requires a change in $X_t$ but not vice-versa --- an asymmetry known as supervenience~\cite{kim1990supervenience}. 

We call a coarse-graining $Z_t$ emergent if, for the purpose of predicting the coarse-grained future $\finfut{Z}{t+1}{L}=(Z_{t+1},\dots,Z_{t+L})$, knowing the past $\past{X}{t}=(\dots,X_{t-1},X_t)$ offers no advantage over knowing only the coarse-grained past $\past{Z}{t}=(\dots,Z_{t-1},Z_t)$, which implies that all relevant information has `moved up' to the coarse-grained level. This idea can be operationalised via the condition of \emph{information closure}~\cite{bertschinger2006information,chang2020information}, which requires that
\begin{equation}\label{eq:info_close}
    I\big( \past{X}{t} ; \finfut{Z}{t+1}{L} \big) 
    - 
    I\big( \past{Z}{t} ; \finfut{Z}{t+1}{L} \big) 
    =0
    \quad \forall L\in\mathbb{N}, t\in\mathbb{Z},
\end{equation}
where $I$ is Shannon's mutual information. 

Information closure can be shown to be equivalent to
\begin{equation}\label{eq:alternative_condition}
\mathbb{E}\left\{
\log \frac{p\big(\finfut{Z}{t+1}{L} | \past{X}{t}\big) }
{p\big(\finfut{Z}{t+1}{L} | \past{Z}{t}\big) } \right\} = 0
\quad \forall L\in\mathbb{N}, t\in\mathbb{Z}.
\end{equation}
As the above expression corresponds to a Kullback-Leibler divergence, this ensures that $p\big(\finfut{Z}{t+1}{L} | \past{X}{t}\big) = 
p\big(\finfut{Z}{t+1}{L} | \past{Z}{t}\big)$ almost surely. 
This implies that the distinct initial conditions of the system that lead to different macroscopic outcomes can be effectively attained from a suitable macroscopic initial condition. Thus, information closure corresponds to self-contained levels of $X_t$ whose regularities are suitable for being described in their own terms~\cite{rosas2024software,barnett2023dynamical,chang2020information}. For this reason, in the sequel informationally closed coarse-grainings will be called `emergent levels' or `macroscopic variables.'

\begin{figure*}[t!]
  \centering
  \if1\compiletikz
  \includetikz{tikz/}{lattices}
  \else
    \includegraphics[width=2\columnwidth]{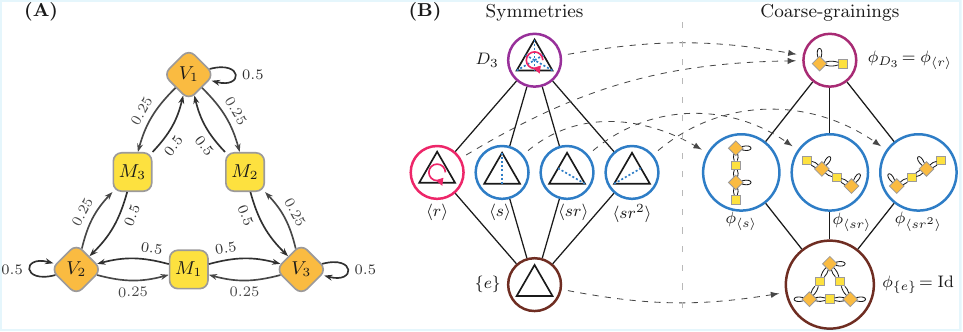}
  \fi 
  \caption{\textbf{Symmetry gives rise to simpler self-contained levels}. \textbf{(A)} Dynamics of a Markov chain over 6 states, which respects symmetries equivalent to a triangle --- known as the dihedral group of order 3, $D_3$. \textbf{(B)} Correspondence between the subgroups of $D_3$ and the lattice of informationally closed levels of the Markov chain, which preserves the partial ordering.}
  \label{fig:elementary}
\end{figure*}

\subsection{Dynamical symmetries}

Consider now a group $G$ acting on $X_t$, so that there exist mappings $g(\cdot):\mathcal{X}\to\mathcal{X}$ for all $g\in G$ such that $(g_1 g_2)\cdot x = g_1\cdot (g_2\gdot x)$. 
A group action induces a partition on $\mathcal{X}$ via `orbits' of the form
\begin{equation}
O_x := \{ x'\in\mathcal{X}: g\gdot x = x', g\in G\}.
\end{equation}
This leads to the coarse-graining $Z_t=\phi_G(X_t)$ given by $\phi_G(x):=O_x$~\footnote{Any coarse-graining of a discrete state space can be seen as arising from a group in this manner. Consider a coarse-graining $\phi$ corresponding to a partition with components $C_i$. Consider $G_i$ the group of permutations of elements in $C_i$, and form $G_\phi=\prod_{i\in \mathcal{I}} G_i$ their cartesian product. Then, the partition corresponds to the orbits of $G_\phi$.}. 
For simplicity, the presentation will focus on finite groups $G$ and discrete random variables.

Under what conditions may the orbits of $G$ generate an informationally closed level? To answer this, let us introduce the Reynolds operator~\cite{reynolds1895iv}
\begin{equation}
    \mathcal{R}\{ \mu \} = \frac{1}{|G|} \sum_{g\in G} g \gdot \mu, 
\end{equation}
which corresponds to averaging over all symmetry-equivalent versions of the dynamics. 
To investigate whether the dynamics respects symmetries, one can average the transition kernel over the group using the Reynolds operator to test if it behaves similarly across symmetry-related microstates. 
When applied on the conditional law over future blocks $K_{t+1}^L(\past{x}{t}) := p(\finfut{x}{t+1}{L}|\past{x}{t})$, the Reynolds operator gives
\begin{equation}\label{eq:reynold_kernel}
    \mathcal{R} \left\{ K_{t+1}^L \right\} (\past{x}{t}) 
    = 
    \frac{1}{|G|^L}\sum_{\finfut{g}{t+1}{L}\in G^L} 
    p(\finfut{g}{t+1}{L}\gdot\finfut{x}{t+1}{L}|\past{x}{t}),
\end{equation}
where $\finfut{g}{t}{L}\gdot\finfut{x}{t}{L} = (g_t\gdot x_t,\dots,g_{t+L}\gdot x_{t+L})$ is the component-wise application of a trajectory of group elements and $G^L$ is the $L$-th Cartesian product of $G$. 
If the resulting averaged kernel is invariant under the group action, then transitions respect the symmetry and the orbit coarse-graining is informationally closed, as shown next.

\begin{theorem}\label{teo:main_result}
The coarse-graining $Z_t=\phi_G(X_t)$ is informationally closed if and only if $\mathcal{R} \{K_{t+1}^L\} (\past{x}{t})$ is invariant to the action of $G$.
\end{theorem}
\begin{proof}
We need to prove that $Z_t = \phi_G(X_t)$ is informationally closed if and only if 
$\mathcal{R} \{ K_{t+1}^L \} ( \past{g}{t}\cdot \past{x}{t}) 
=
\mathcal{R} \{K_{t+1}^L\} (\past{x}{t})$ for all $\past{g}{t}$ and all $t\in\mathcal{Z}$ and $L\in\mathbb{N}$. 
If $Z_t = \phi_G(X_t)$ is given by the group orbits, then
\begin{align}
    p(\finfut{z}{t+1}{L} |\past{x}{t})
    &= 
    \!\!\!\!\!
    \sum_{\finfut{x}{t+1}{L}\in \phi_{G^L}^{-1}(\finfut{z}{t+1}{L})}
    \!\!\!\!\!
    p( \finfut{x}{t+1}{L} | \past{x}{t} ) \\
    &\overset{(a)}{=} 
    \left|\phi_{G^L}^{-1}(\finfut{z}{t+1}{L})\right| \mathcal{R} \left\{K_{t+1}^L\right\}(\past{x}{t}),
\end{align}
where (a) is a consequence of the orbit-stabiliser theorem. Given that $p(\finfut{z}{t+1}{L} |\past{x}{t}) = p(\finfut{z}{t+1}{L} |\past{g}{t}\gdot\past{x}{t})$ for all $\finfut{g}{t+1}{L}\in G^L$ is equivalent to \autoref{eq:alternative_condition}, and hence also to \autoref{eq:info_close}, this proves the desired result.
\end{proof}

This general result has a simple yet powerful consequence for Markov processes. 
One says that a Markov process $X_t$ is \emph{dynamically equivariant} if there exists a group action such that
\begin{equation}
    p(g\gdot x_{t+1}|g\gdot x_t)=p(x_{t+1}|x_t)
\end{equation}
for all $x_t,x_{t+1}\in\mathcal{X},g\in G$. 
Each orbit of this action corresponds to a set of microstates that are dynamically indistinguishable under the symmetry. One can, therefore, define a macroscopic variable $Z_t=\phi_G(X_t)$ that records only which orbit the system is in. 
The next result is that this orbit variable is informationally closed~\footnote{Note that \autoref{cor:equivariance} shows that equivariance is sufficient for emergence, but not necessary. A necessary condition is given by Markov lumpability, %~\cite[Ch.~6.3]{kemeny1969finite}, 
which can be seen as a local symmetry using the formalism of fibration symmetries.}%~\cite{makse2025symmetries}.}.

\begin{corollary}\label{cor:equivariance}
    If a Markov process  $X_t$ is dynamically equivariant, then $Z_t=\phi_G(X_t)$ is informationally closed.
\end{corollary}
\begin{proof}
   In the Markov case $K_{t+1}^L$ only depends on $x_t$, so one can write $K_{t+1}^L(x_t)$. 
   Under those conditions, 
   \begin{align}
   |G&|^L \mathcal{R} \{ K_{t+1}^L \} ( g_t\gdot x_t)
   \!=
   \!\!\!\!\!\!
   \sum_{\finfut{g}{t+1}{L}\!\in G^{L}} 
   \!
   \prod_{\tau=t}^{t+L}
   p(g_{\tau+1}\!\gdot x_{\tau+1}|g_\tau\gdot x_\tau) \\
   &=
   \sum_{\finfut{g}{t+1}{L}\in G^L} 
   \prod_{\tau=t}^{t+L}
   p\big((g_t^{-1}\gdot g_{\tau+1})\gdot x_{\tau+1}|(g^{-1}_t\gdot g_\tau)\gdot x_\tau) \\
   &=
   \sum_{\finfut{g}{t+1}{L}\in G^L} 
   p\big((g^{-1}_t\!\cdot \finfut{g}{t+1}{L})\cdot \finfut{x}{t+1}{L}| x_t) \\
   &=
   |G|^L \mathcal{R} \{ K_{t+1}^L \} (x_t),
   \end{align}
   where $g^{-1}_t\!\cdot \finfut{g}{t+1}{L} = (g_t^{-1}\!\cdot g_{t+1},\dots,g_t^{-1}\!\cdot g_{t+L})$ denotes the component-wise product.
\end{proof}

An elementary example of an equivariant Markov chain is given in \autoref{fig:elementary}A. 
When applied to deterministic dynamical systems (which become Markov processes when considering random initial conditions~\cite{rosas2018information}), this corollary generalises approaches to factor out symmetries of equivariant maps~\cite{field1980equivariant,golubitsky2012singularities}.

\subsection{Hierarchical organisation}

So far, we have seen how a single symmetry yields one emergent macroscopic description. However, most systems exhibit many symmetries, corresponding to groups with non-trivial algebraic structures. In particular, they can have a family of proper subgroups, which naturally introduces a hierarchy of emergent levels, as shown by the next corollary. 

\begin{corollary}\label{cor:subgroups}
    Let $X_t$ be a Markov process that is equivariant with respect to $G$, and $H$ be a subgroup of $G$. If $Z_t=\phi_G(X_t)$ and $Z'_t=\phi_H(X_t)$, then $Z_t$ is informationally closed with respect to both $X_t$ and $Z'_t$.
\end{corollary}
\begin{proof}
The first part of the result follows directly from the fact that if $X_t$ is equivariant with respect to $G$, it is also equivariant with respect to any subgroup $H$, and hence \autoref{cor:equivariance} also applies. 
To show that $Z_t$ is informationally closed with respect to $Z'_t$, note that the fact that $H$ is a subgroup of $G$ implies that $Z'_t$ is a finer coarse-graining of $X_t$ than $Z_t$, and hence there exists a coarse-graining mapping $\phi'$ such that $Z_t=\phi'(Z'_t)$. Then, the data processing inequality implies that
\begin{align}
    0
    &\leq 
    I(\past{Z}{t}';\finfut{Z}{t+1}{L}) - I(\past{Z}{t};\finfut{Z}{t+1}{L})  \nonumber\\
    &\leq
    I(\past{X}{t};\finfut{Z}{t+1}{L}) - I(\past{Z}{t};\finfut{Z}{t+1}{L}) 
    = 
    0.
\end{align}
\end{proof}
\autoref{cor:subgroups} implies that an equivariant Markov process has a whole hierarchy of emergent levels, whose organisation reflects the lattice structure of the subgroups of $G$ (see \autoref{fig:elementary}B). 
Concretely, if $H\subseteq G$ then the orbits under $G$ are unions of orbits under $H$, so $Z_t=\phi_G(X_t)$ is coarser (i.e. less detailed) than $Z'_t=\phi_H(X_t)$. 
This implies that there is a homomorphism between the lattice of subgroups of $G$ and the informationally closed coarse-grainings, which preserves the partial ordering but can result in a simpler arrangement. 
In this hierarchy, the most coarse-grained level corresponds to the group of all symmetries ($G$), and the finest level to the trivial subgroup with single element orbits ($\{e\}$).

\begin{figure*}[t!]
  \centering
  \if1\compiletikz
  \includetikz{tikz/}{hopfield}
  \else
    \includegraphics[width=2\columnwidth]{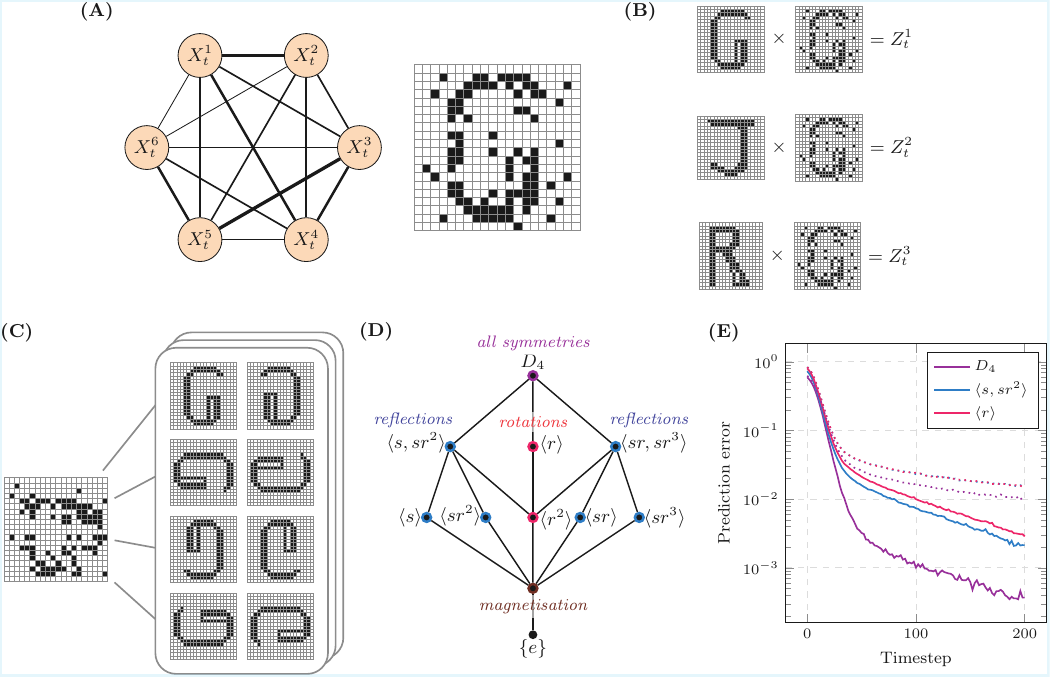}
  \fi 
  \caption{\textbf{Equivariant symmetries in a Hopfield network}. 
  \textbf{(A)} A Hopfield network is a recurrent neural network trained to retrieve a given set of patterns. 
  \textbf{(B)} The dot products between each candidate pattern and the current state of the network are equivariant to the neural dynamics, providing natural order parameters known as Mattis magnetisation. 
  \textbf{(C)} Training the network over letters and all their rotations and reflections induce additional equivariances to the dihedral group $D_4$. 
  \textbf{(D)} The subgroups of $D_4$ correspond to different informationally closed coarse-grainings. For example, factoring all symmetries out yields a representation that captures only letter identity, while factoring by the rotation subgroup ($\langle r\rangle$) preserves letter identity and whether it is reflected. 
  \textbf{(E)} Numerical results demonstrate that these symmetry-derived coarse-grainings (solid lines) yield substantially better self-prediction than arbitrary coarse-grainings of the same size (dotted lines matched by color), confirming that symmetry-based emergent variables capture all dynamically relevant information.}
\label{fig:hopfield}
\end{figure*}

\section{Case study 1: Hopfield networks} 
\label{sec:hopfield}

Let's illustrate the relationship between emergence and symmetry in Hopfield networks, a paradigmatic model of memory storage and retrieval~\cite{amari1972learning,hopfield1982neural} which has played a longstanding role linking machine learning and physics~\cite{hinton1986learning,krotov2023new,aguilera2025explosive}. 
We will see how the network dynamics' and the symmetries in the stored patterns give rise to macroscopic order parameters whose dynamics are informationally closed.

A Hopfield network is a recurrent neural network where each neuron connects to all others, and the connection weights are designed so that certain patterns are attractors~\cite{amit1989modeling,gerstner2014neuronal} (see \autoref{fig:hopfield}A). 
The state of the $i$-th neuron $X_t^i\in\{-1,1\}$ updates according to 
\begin{equation}\label{eq:hopfield_dyns}
    \mathbb{P}\big\{X_{t+1}^i=1|X_t^1,\ldots,X_t^n\big\}
    %=\mathbb{P}\big\{X_{t+1}^j=1|h_t^j\big\}
    = F\Big(\beta \sum_{j=1}^n w_{i,j} X_t^j \Big)~,     
\end{equation}
where $F(x) = (1+\tanh{x})/2$ is the activation function, $\beta$ is a parameter that regulates the stochasticity of the updates, and $w_{i,j}$ is the strength of the synapse from the $j$-th neuron. 
The dynamics of the whole network $X_t=(X_t^1,\ldots,X_t^n)$ are given by 
$p(x_{t+1} |x_t) = \prod_j p\big(x_{t+1}^j|x_t\big)$.

To store and retrieve patterns, Hopfield networks typically use Hebbian learning~\cite{hebb1949organization}, which strengthens the connection between neurons that tend to co-activate. 
Given $m$ patterns $q_\mu=(q_\mu^1,\ldots,q_\mu^n) \in\{-1,1\}^n$, this suggests establishing synaptic weights according to
\begin{equation}\label{eq:hebb}
    w_{i,j} = 
    \begin{cases}
        \frac{1}{n}\sum_{\mu} q_\mu^i q_\mu^j 
        &\quad \text{if }i\neq j,\\
        0&\quad \text{otherwise.}
    \end{cases}
\end{equation}
These synaptic weights make each pattern $q_\mu$ a stable fixed point of the dynamics (for large $n$), so the network naturally relaxes toward one of these attractors~\cite[Ch.~17.2]{gerstner2014neuronal}. 
Under these weights, the input potential to the $i$-th neuron can be expressed as
\begin{equation}
    \sum_{j=1}^n w_{i,j} X_t^j 
    = \frac{1}{n}\sum_{\mu=1}^m q_\mu^i \sum_{j=1}^m q_\mu^j X_t^j,
\end{equation}
making the so-called Mattis magnetisation (\autoref{fig:hopfield}B)
\begin{equation}
    Z_t^\mu = \frac{1}{n} (q_\mu \gdot X_t) = \frac{1}{n}\sum_{j=1}^n q_\mu^j X_t^j
\end{equation}
a natural order parameter of the dynamics~\cite[Ch.4.4]{amit1989modeling} (see also~\cite{barra2018new,agliari2023parallel}). 
The vector of Mattis magnetisations $Z_t = (Z_t^1,\dots,Z_t^m)$ tracks the degree of similarity of the current pattern $X_t$ with each of the candidate patterns (a typical run starts with all magnetisations being low and ends with one dominating over the rest).

Let us now investigate the symmetries of this system. For this, let's build a group $G$ containing all the permutations of neurons $\sigma$ that leave the magnetisation unchanged, so that
$\sigma\gdot x_t:=(x_{\sigma(1)}^t,\dots,x_{\sigma(n)}^t)$ satisfies
\begin{equation}\label{eq:dotdot}
\big( q_\mu\gdot (\sigma\gdot x) \big)
= 
( q_\mu\gdot  x )
\quad
\forall \mu=1,\mydots,m,
 x\in\{-1,1\}^n.
\end{equation}
Under this construction, the magnetisation corresponds to the orbits of the symmetry (i.e. $Z_t=\phi_G(X_t)$). 
Moreover, one can show that
\begin{equation}\label{eq:prob_glob}
     p(\sigma\gdot x_{t+1} |\sigma\gdot x_t) 
     = 
     \prod_{j=1}^n  p\big(x_{t+1}^{\sigma(j)}|g\gdot x_t\big)
     = 
     \prod_{j=1}^n  p\big(x_{t+1}^j|x_t\big),
\end{equation}
where the second equality results from $p(x_{t+1}|x_t)$ depending only on the magnetisation $z_t=\phi_G(x_t)$. 
This dynamical equivariance implies, thanks to \autoref{cor:equivariance}, that $Z_t$ is informationally closed, confirming the role of the magnetisation as a macroscopic variable. 
This means that the process of selecting which pattern to retrieve takes place at a macroscopic level ($Z_t$), being (conditionally) independent of the instantiation at the level of individual neurons ($X_t$). 
This implies, for instance, that the identity of the pattern to be retrieved (which is specified by $Z_t$ for $t\gg1$) depends only on the macroscopic initial magnetisation $Z_0$, and not on the additional information provided by $X_0$.

To further study the effect of symmetries on Hopfield networks, let us consider a square arrangement of neurons encoding all reflections and $90^\circ$ rotations of $k_0$ asymmetric letters, leading to $m=8k_0$ patterns (see \autoref{fig:hopfield}C). This choice of patterns causes the synaptic weights, and thus also the neural dynamics (due to \autoref{eq:hopfield_dyns}), to have additional symmetries corresponding to the dihedral group $D_4$ --- which is generated by the $90^\circ$ rotation ($r$) and the horizontal reflection ($s$). 
In this context, the coarse-graining $Z'_t = \phi_{D_4}(Z_t)$ corresponds to the evidence for each letter irrespective of their rotation or reflection.

Additionally, $D_4$ provides 8 proper subgroups (see \autoref{fig:hopfield}D): three subgroups of order 4 and five subgroups of order 2~\footnote{Out of the three groups of order 4, one is a cyclic group generated by $r$ corresponding to the rotations, and two groups generated by reflections isomorphic to Klein's group $\mathbb{Z}_2\times\mathbb{Z}_2$.
The five groups of order 2 are generated by $r^2$, $s$, $sr$, $sr^2$, and $sr^3$.}. 
Each subgroup generates a distinct informationally closed coarse-graining: for instance, the cyclic group of order 4 factors out rotational information, while the two subgroups isomorphic to the Klein group $\mathbb{Z}_2\times\mathbb{Z}_2$ factor out information about certain reflections. Hence, the subgroups of order 4 disentangle rotations and reflections, illuminating the compositional structure of the corresponding operations.

The theory predicts that the information provided at each of these levels should be optimal for predicting itself in the future --- e.g., the information given by $\phi_{D_4}(Z_t)$ should be sufficient to optimally predict what letter will be retrieved, irrespective of its orientation or rotation. 
To further explore this, Hopfield networks of $20\times 20$ neurons were trained on the rotations and reflections of the letters $G$, $J$, and $R$, and used the state of $\phi_{H}(Z_t)$ for different subgroups $H$ at various timepoints to predict itself after 250 timesteps (for details, see Appendix~\autoref{app:hopfield}). 
Results show that the ability of informationally-closed levels to predict their own future is substantially higher than randomly selected coarse-grainings of the same cardinality, which are typically not informationally closed and hence leak information into the micro-level~(see \autoref{fig:hopfield}E). 
This confirms that emergent variables derived from symmetries are not only conceptually natural but also practically desirable for prediction.

\section{Hierarchical abstractions}

After focusing on the objective dynamics of a system, we now turn to an observer who only has access to noisy measurements and maintains probabilistic beliefs about its state. The aim is to show that the same symmetries that generate emergent macrostates give rise to an analogous hierarchy of belief states, which we will call \emph{abstractions}.

\subsection{Multi-level latent models}

Let us investigate the relationship between the hierarchical organisation of a process and measurements of it by focusing on hidden Markov models (HMM). 
In an HMM, a latent (i.e., unobservable) Markov process $X_t$ evolves according to a transition kernel 
$p(x_{t+1}| x_t)$, while observations $Y_t$ are generated through an emission kernel $p(y_t|x_t)$. 
This yields a joint distribution of latent states and measurements of the form
\begin{equation}\label{eq:HMM}
    p(\finfut{x}{0}{t+1},\finfut{y}{0}{t})
    =
    p(x_0)
    \prod_{\tau=0}^{t} 
    p( x_{\tau+1}| x_{\tau})
    p( y_{\tau}|x_{\tau} ) .
\end{equation}

Let us now consider a scenario where (i) the hidden process is dynamically equivariant and (ii) the measurements respect this symmetry. Technically, this requires two group actions from the same group, one over $\mathcal{X}$ and another over $\mathcal{Y}$, satisfying the following conditions:
\begin{subequations}\label{eq:sym_HMM}
\begin{align}
    p(g\gdot x_{t+1}|g\gdot x_t) 
    &= p\big(x_{t+1}|x_t\big), \label{eq:sym_HMMa}\\ 
    p(g\gdot y_{t}|g\gdot x_t)
    &= p\big(y_{t}|x_t\big), \quad \text{and} \label{eq:sym_HMMb}\\
    p(g\gdot x_0)
    &=p(x_0). \label{eq:initial_sym}
\end{align}
\end{subequations}
Systems satisfying these properties will be called \emph{equivariant HMMs}. 
The group actions can be used to build coarse-grained variables
\begin{equation}\label{eq:simp}
Z_t=\phi_G(X_t)     
\quad\text{and}\quad 
V_t=\psi_G(Y_t)
\end{equation}
associated with the orbits arising from these symmetries, which form themselves another HMM (proof in App.~\ref{app:proof_teo_HMM}).
\begin{theorem}\label{teo:HMM}
    If $(X_t,Y_t)$ is an HMM satisfying \autoref{eq:sym_HMM}, then $(Z_t,V_t)$ also form an HMM. 
\end{theorem}
In other words, if the underlying physical process and the observation model satisfy symmetries, then the orbits of the microscopic variables and the observer's measurements form a new HMM nested inside the original one. 
This allows observers to operate entirely at the emergent level, as discussed in the next subsection.

\begin{figure*}[t!]
  \centering
  \if1\compiletikz
  \includetikz{tikz/}{hierarchical_HMM}
  \else
    \includegraphics[width=1.98\columnwidth]{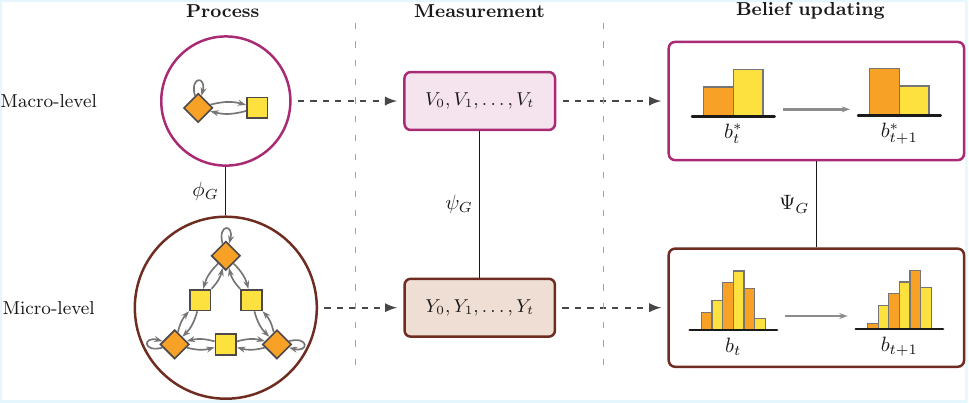}
  \fi 
  \caption{\textbf{Hierarchical Bayesian belief updating}. A process with equivariant symmetries has informationally closed levels (\textit{left}). If measurements respect the same symmetries, then there is a correspondence between the various levels of the process and measurements of different degrees of resolution (\textit{centre}). This hierarchical structure of processes and measurements gives rise to Bayesian beliefs that can be updated autonomously at various scales, without the need of accounting for the information from levels below (\textit{right}).}
  \label{fig:hierarchical_beliefs}
\end{figure*}

\subsection{Hierarchical belief updating}

HMMs let us now study how (noisy or partial) measurements inform us about latent processes. 
The optimal estimation of $X_t$ based on a sequence of measurements $\past{Y}{t}$ is given by \emph{Bayesian beliefs}, a Markov process $B_t\coloneq\mathbb{E}\{\mathds{1}_{X_t}|\past{Y}{t}\}$ where $\mathds{1}_{X_t}\coloneq\mathds{1}\{X_t=x_t\}$ is a one-hot encoding of $X_t$~\cite{sarkka2023bayesian}. When conditioned on a realisation $\{\past{Y}{t}=\past{y}{t}\}$, Bayesian beliefs take the form $b_{t}(x_t)=p(x_{t}|\past{y}{t})\in\mathbb{R}^{|\mathcal{X}|}$, which can be efficiently updated via standard Bayesian filtering as 
\begin{equation}\label{eq:bayes_matrix}
    b_{t} 
    = \text{Norm}\big( \diag( L_{y_t}) T_t b_{t-1} \big) 
    \coloneqq U\big(y_t,b_{t-1}\big).
\end{equation}
Above, $T_t\in\mathbb{R}^{|\mathcal{X}|\times|\mathcal{X}|}$ is a transition probability matrix with components $[p(x_{t}|x_{t-1})]_{x_{t},x_{t-1}\in\mathcal{X}}$, $L_{y_t}$ is a $|\mathcal{X}|$-dimensional likelihood vector with components $[L_{y_t}]_{x_t\in\mathcal{X}} = p(y_t|x_t)$, and $\text{Norm}(v)=v/\sum_j v_j$. 
Intuitively, $T_t$ propagates the belief through the dynamics and then $\diag( L_{y_t})$ re-weights the result according to the evidence brought by $y_t$.

One can also build beliefs about arbitrary coarse-grained versions of the latent process $W_t = f(X_t)$, resulting in 
$B^f_t:= \mathbb{E}\{\mathds{1}_{W_t}|\past{Y}{t}\}$ and 
$b^f_t:= p(w_t|\past{y}{t})$. These beliefs satisfy $B^f_t = P_f B_t$, where $P_f\in \{0,1\}^{|\mathcal{W}|\times|\mathcal{X}|}$ is a matrix with components $[P_f]_{w,x}=\mathds{1}\{w=f(x)\}$. 
However, most such coarse-grainings result in $B^f_t$ being non-Markovian, which forbids the use of update mechanisms analogous to \autoref{eq:bayes_matrix}. 
Coarse-grainings associated with symmetries circumvent these issues, resulting in Markovian beliefs that can be updated efficiently.

\begin{corollary}\label{cor:cool}
    Under the conditions of \autoref{teo:HMM}, the belief $b_t^*:= p(z_t|\past{v}{t})$ can be updated via $b^*_t=U(v_t, b^*_{t-1}) $ with $U$ the update operator defined in \autoref{eq:bayes_matrix}. Moreover, 
    \begin{equation}\label{eq:cool_beliefs}
        b^*_t = P_{\phi_G} \mathbb{E} \big\{ B_t | \past{V}{t}=\past{v}{t} \big\}.
    \end{equation}
\end{corollary}

\begin{proof}
    The beliefs $b^*_t$ satisfy $b^*_t = U(v_t,b^*_{t-1})$ because, thanks to \autoref{teo:HMM}, $(Z_t,V_t)$ is an HMM. Furthermore,
    \begin{equation}\label{eq:very_nice}
    b^*_t 
    = 
    \!\!
    \sum_{\past{y}{t}\in\psi^{-1}(\past{v}{t})}
    \!\!
    p(\past{y}{t}|\past{v}{t}) p(z_t|\past{y}{t})
    = \mathbb{E} \big\{ B'_t | \past{V}{t}=\past{v}{t} \big\},
    \end{equation}
with $B'_t=\mathbb{E}\{\mathds{1}_{Z_t}|\past{Y}{t}\}=P_{\phi_G}B_t$.
\end{proof}

Just as symmetries induce emergent macrostates in the physical system, \autoref{cor:cool} reveals they also induce simpler belief states --- henceforth called \emph{abstract beliefs} --- for doing inference over informationally closed levels of the latent process (see \autoref{fig:hierarchical_beliefs}). 
These abstract beliefs are built via \autoref{eq:cool_beliefs}, which mixes micro-level beliefs corresponding to measurement sequences that differ only by symmetry transformations:
\begin{align}
    \Psi_G\{b_t[\past{y}{t}]\} 
    :=& 
    \mathbb{E}\big\{ \mathds{1}_{Z_t}|V_t=\psi_G(\past{y}{t}) \big\} \\
    =&
    P_{\phi_G}
    \Bigg(
    \sum_{\past{y}{t}': \psi_G(\past{y}{t}')=\psi_G(\past{y}{t})}
    \!\!\!\!\!\!
    c_{\past{y}{t}'} b_t[\past{y}{t}'], \label{eq:mapping_psi}
    \Bigg),
\end{align}
where the notation $b_t[\past{y}{t}] = p(x_t|\past{y}{t})$ highlights the dependency of the belief on the measurement sequence $\past{y}{t}$, and $c_{\past{y}{t}'}$ are mixing coefficients given by
\begin{equation}
    c_{\past{y}{t}'} = p(\past{y}{t}'|\psi_G(\past{y}{t})) = \frac{p(\past{y}{t}')}{\sum_{\past{y}{t}': \psi_G(\past{y}{t}')=\psi_G(\past{y}{t})} p(\past{y}{t}')}.
\end{equation}  
This mapping can be seen as a two-step process: it first mixes the micro-level beliefs $b_t[\past{y}{t}']$ in the equivalence class $\past{y}{t}'\in\psi^{-1}(\past{y}{t})$, and then marginalises the mixture via $P_{\phi_G}$. 
It is important to note that while \autoref{eq:mapping_psi} holds for the beliefs of any coarse-graining of the latent process $p(w_t|\past{y}{t})$ for $W_t=f(X_t)$, these generally cannot be updated as in \autoref{eq:bayes_matrix}. 
It is the possibility of efficiently updating abstract beliefs by direct Bayesian filtering what makes these special.

While \autoref{eq:mapping_psi} states how to build abstract beliefs from micro-level ones, this generally does not yield a well-defined mapping between individual beliefs. A sufficient condition for such mapping to exist is given next.
\begin{corollary}\label{cor:of_the_cor}
    Under the conditions of \autoref{cor:cool}, if
    \begin{equation}\label{eq:condition_strong}
    I(Z_t;\past{Y}{t}) - I(Z_t;\past{V}{t}) = 0
    \end{equation}
    then $b^*_t = P_{\phi_G} b_t$. 
\end{corollary}
\begin{proof}
    The above condition is equivalent to
    \begin{equation}
    \mathbb{E}\left\{\log \frac{p(Z_t|\past{Y}{t})}{p(Z_t|\past{V}{t})}\right\} 
    = 
    0,
    \end{equation}
    which implies that $p(z_t|\past{v}{t})=p(z_t|\past{y}{t})$ for all $\psi_G(\past{y}{t}) = \past{v}{t}$ almost surely. Using \autoref{eq:very_nice}, this leads to
    \begin{equation}
    b^*_t 
    = 
    p(z_t|\past{y}{t})
    \!\!
    \sum_{\past{y}{t}\in\psi^{-1}(\past{v}{t})}
    \!\!
    p(\past{y}{t}|\past{v}{t}) 
    = P_{\phi_G} b_t.
\end{equation}
\end{proof}
\autoref{eq:condition_strong} states that the coarse-grained measurement $\past{V}{t}$ contains all the information about $\finfut{Z}{t}{L}$ that exists in the full measurement $\past{Y}{t}$, being analogous to information closure at the belief level~\footnote{In fact, this result is a generalisation of the result presented in \cite{rosas2024software} showing that information closure implies computational closure.}. 
An equivalent condition is %in terms of symmetries is
\begin{equation}
    p(z_t|\past{y}{t}) = p(z_t|g_0\gdot y_0,\dots,g_t\gdot y_t), 
    \quad
    \forall g_k\in G.
\end{equation}
If these conditions hold, then abstract beliefs are just a marginalisation of micro-level beliefs. This, in turn, ensures that the following diagram commutes:
\begin{figure}[h!]
  \centering
  \if1\compiletikz
  \includetikz{tikz/}{commuting_diagram}
  \else
    \includegraphics[width=0.82\columnwidth]{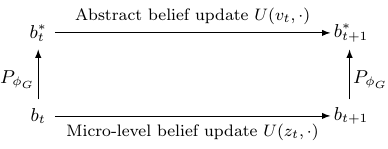}
  \fi 
\end{figure}

\noindent
This guarantees that Bayesian updating at different levels of abstraction are fully compatible: updating beliefs at a coarse level corresponds to marginalising the fine-grained beliefs, and vice versa. 
In contrast, when \autoref{eq:condition_strong} does not hold then $I(Z_t;\past{Y}{t}) > I(Z_t;\past{V}{t})$, meaning that the beliefs $p(z_t|\past{y}{t})=P_{\phi_G} b_t$ are more informative about $Z_t$ than $b^*_t$~\footnote{Note that, due to the tower property of conditional expectations, \autoref{eq:cool_beliefs} guarantees that both $b^*_t$ and $P_{\phi_G}b_t$ always agree on their mean, but the variance of $b_t^*$ will generally be larger.}. However, there may not exist an efficient way to update $p(z_t|\past{y}{t})$ without updating $b_t$ first (as $(Z_t,Y_t)$ may not be an HMM), making $b_t^*$ a less informative but more tractable alternative.

In summary, these results show that beliefs arising from equivariant HMMs can be updated at different levels of granularity, depending on the underlying symmetries (see \autoref{fig:hierarchical_beliefs}). Updating $b_t=p(x_t|\past{y}{t})$ provides estimates of system's state in full resolution, but their updating and storage can involve large computational costs --- especially when the phase space $\mathcal{X}$ is large. In contrast, updating $b^*_t=p(z_t|\past{v}{t})$ provides estimates about a macroscale description of the system, resulting in lower dimensional beliefs whose updating and storage can be substantially cheaper.

\begin{figure*}[!t]
  \centering
  \if1\compiletikz
  \includetikz{tikz/}{ehrenfest}
  \else
    \includegraphics[width=2\columnwidth]{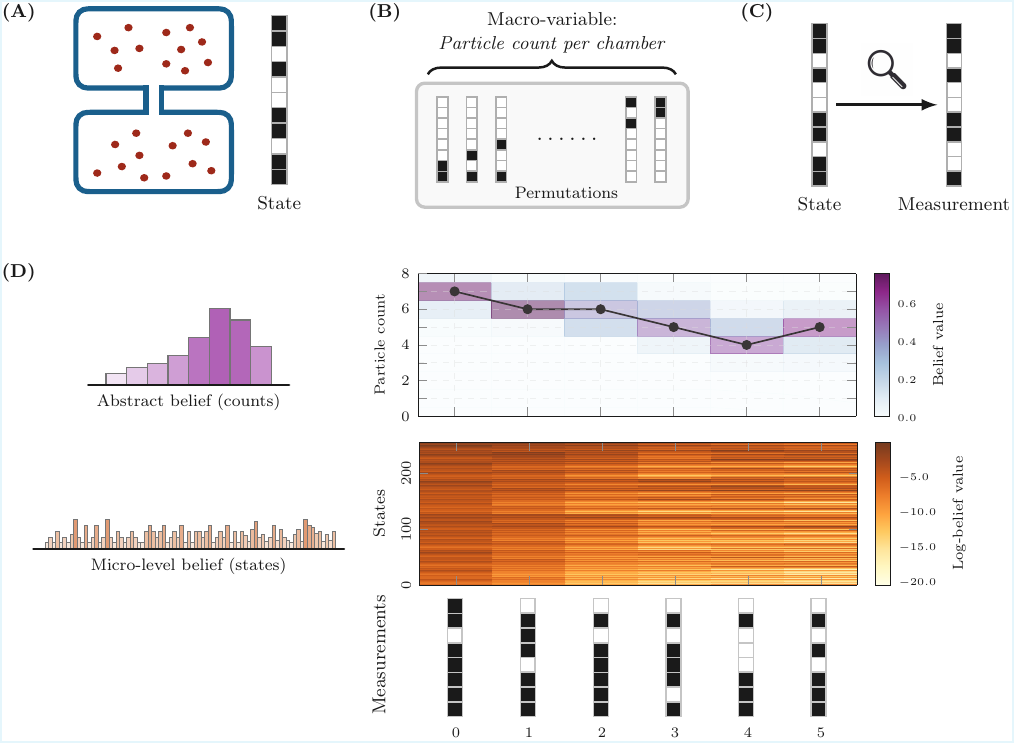}
  \fi 
  \caption{\textbf{Abstract beliefs in the Ehrenfest diffusion model.} 
  \textbf{(A)} The model considers $n$ particles contained in two connected chambers. The micro-level state is a binary vector that specifies in which container is each particle. 
  \textbf{(B)} The system's dynamics are equivariant with respect to permutations among the particles, giving rise to equivalence classes corresponding to the particle count per chamber. 
  \textbf{(C)} Measurements are noisy observations of each particle’s location, obtained via a binary symmetric channel.
  \textbf{(D)} Updating of micro-level beliefs about the microstate and abstract beliefs about the particle count for $n=8$. 
  The dimensionality of abstract beliefs grows linearly with $n$, while for micro-level beliefs it grows exponentially.
  }
  \label{fig:Ehrnfest}
\end{figure*}

\section{Case study 2: Ehrenfest's model} 

Let us illustrate these results on Ehrenfest's diffusion model~\cite{ehrenfest1907zwei}, a classic model from statistical mechanics that describes the behaviour of gas in a container made of two interconnected chambers. 
We will see how permutation symmetry at the microstate gives rise to the number of particles per box as a natural order parameter, over which efficient inference is possible.

The model describes the state of $n$ gas molecules via a binary vector $X_t = (X_t^{1},\dots,X_t^{n})$, where $X_t^{k} \in \{0,1\}$ determines in which of the two chambers the $k$-th molecule is located (see \autoref{fig:Ehrnfest}A). 
The model assumes that at most one molecule may switch between chambers at each time point, giving rise to the following dynamics:
\begin{align}\label{eq:probs_ehren}
    p(x_{t+1}|x_t) 
    =
    \begin{cases}
    1 - q  \quad &\text{\small if } x_{t+1}=x_t, \\
    q/n &\text{\small if } \|x_{t+1}-x_t\| = 1,  \\
    0 \quad &\text{\small otherwise},
    \end{cases}
\end{align}
where $q$ is the probability of a particle switching and $\|x-x'\|$ is the Hamming distance between $x$ and $x'$. 
Let us also consider measurements $Y_t = (Y_t^{1},\dots,Y_t^{n})$, where $Y_t^k\in\{0,1\}$ is an estimate of $X_t^k$ obtained from a binary symmetric channel with crossover probability $r$, so that 
\begin{align}\label{eq:measurement_erhenfest}
    p(y_t|x_t) 
    &= 
    \prod_{k=1}^n 
    r^{\mathds{1}\{x_t^k \neq y_t^k\}}
    (1-r)^{\mathds{1}\{x_t^k = y_t^k\}}.
\end{align}
From these measurements one can build beliefs about $X_t$ of the form $b_t=p(x_t|\past{y}{t})$. 
Unfortunately, updating and storing such beliefs is computationally unfeasible for systems with more than a few particles: the dimensionality of $b_t$ is $|\mathcal{X}|=2^n$, so the calculation of one pass of  \autoref{eq:bayes_matrix} requires $\mathcal{O}(2^{2n})$ operations. For example, performing a single belief update on a system with $n=40$ particles with a supercomputer running at $10^{18}$ FLOPS would take around fourteen days, and with $n=50$ particles it would require around forty thousand years.

Let us now use symmetries to build and update coarse-grained beliefs about macroscopic properties. Given that gas particles are identical in this model, the dynamics of $X_t$ are equivariant with respect to the permutation group $S_n$ (see \autoref{fig:Ehrnfest}B). Indeed, by defining 
\begin{equation}
    \sigma\gdot x_t = \big(x_t^{\sigma(1)},\dots,x_t^{\sigma(n)}\big)
    \quad
    \text{for }
    \sigma\in S_n,
\end{equation} 
one can use \autoref{eq:probs_ehren} to verify that $p(\sigma\gdot x_{t+1}|\sigma\gdot x_t)=p(x_{t+1}|x_t)$ holds. 
The resulting orbits are given by $Z_t=\phi_{S_n}(X_t) = \sum_{k=1}^n X_t^k$, which keeps track of the dynamics of the number of particles in each chamber. 
Given the dynamics of the Ehrenfest model, $Z_t$ can be seen to follow a Markovian death-birth process --- because only one particle can flip per timestep, so $Z_t$ increases or decreases by one with probability proportional to the number of particles on each side.

In addition to $X_t$ satisfying dynamical equivariance, it is direct to see that $p(\sigma\cdot y_t|\sigma\cdot x_t) = p(y_t|x_t)$ also holds. By further assuming a permutation-invariant initial condition for $X_0$, this makes $(X_t,Y_t)$ to satisfy \autoref{eq:sym_HMM}. Thus, one can use the previous results to generate and update beliefs about $Z_t$ using coarse-grainings of $Y_t$. Specifically, by defining $V_t = \sum_{k=1}^n Y_t^k$, 
\autoref{teo:HMM} shows that $Z_t$ and $V_t$ also constitute an HMM, so \autoref{cor:cool} guarantees that one can use $V_t$ to track the evolution of $Z_t$ via the abstract belief $b_t^*=p(z_t|\past{v}{t})$. 

Crucially, the abstract beliefs $b_t^*$ are $n+1$ dimensional and hence are much more efficient to compute and update than the micro-level beliefs $b_t$ (see \autoref{fig:Ehrnfest}D). 
Following the above example, updating the abstract beliefs for systems with $n=40$ or $n=50$ particles using a supercomputer running at $10^{18}$ FLOPS would take no more than a few femptoseconds ($10^{-15}$s). 
That said, note that \autoref{eq:condition_strong} is not satisfied in this scenario, even when $X_0$ is initialised with a permutation-invariant prior. 
Moreover, numerical evaluations have confirmed the validity of the mapping between micro-level and abstract beliefs established in \autoref{eq:mapping_psi}.

In addition to the dynamics of the number of particles (which is the coarsest informationally closed level of the system), the Ehrenfest model has an intricate hierarchy of intermediate informationally closed levels corresponding to the subgroup lattice of $S_n$. 
These intermediate levels correspond to distinctions within the configurations with an equal number of particles in each chamber, where disjoint groups of $M$ molecules become new units $z^k_t\in\{0,\dots,M\}$. Indeed, each subgroup corresponds to treating certain particles as indistinguishable (i.e., allowing permutations only within specific blocks), which produces progressively coarser descriptions of the system in a similar way as renormalisation groups merge microscopic degrees of freedom into larger effective units. 
Thanks to \autoref{cor:cool}, all these levels can also be tracked via coarse-grained beliefs, which can be updated using coarse-grained measurements that sum the estimation of particles in the corresponding blocks.

\section{Discussion}

This paper combines ideas from information theory, group theory, and statistical mechanics to propose symmetry as a fundamental mechanism at the origin of hierarchical emergence. 
Indeed, symmetries that commute with the dynamics can always be factored out to reveal simpler self-contained levels of a process. 
This seemingly simple intuition has far-reaching implications thanks to the power of group theory, which guarantees not just one but a whole range of emergent levels arranged in accordance with the subgroup lattice. 
Furthermore, emergent levels displayed by objective processes are reflected in Bayesian beliefs about them that can be hierarchically updated at various levels of resolution. 
This, in turn, opens important opportunities for investigating high-dimensional systems that are only partially observable, as tracking Bayesian beliefs targeted to macroscopic properties can make feasible inferential problems that would otherwise be computationally intractable.

In the objective domain, the results presented here illuminate the origins of hierarchical emergence, contributing to better understand not only \emph{when} emergence happens but also \emph{how} it does so --- backtracking from patterns of behaviour to mechanisms~\cite{rosas2022disentangling}. 
This complements ongoing efforts to understand other flavours of emergence focused on enhanced controllability~\cite{hoel2025causal,jansma2025engineering} or whole-part relationships~\cite{varley2022flickering,tolle2024evolving,pigozzi2025associative}. 
That said, whether symmetry is relevant for these other kinds of emergence is an interesting open question. 
Renormalisation group theory also provides an explanation for why emergence occurs near criticality through scale-invariance, and more generally through effective field theories that filter irrelevant degrees of freedom~\cite{batterman2001devil}. 
Symmetry plays a different role in renormalisation and effective field theory, determining the compatible operators and providing constraints to the emergent laws. In contrast, symmetries in this work refer to dynamics and directly determine the form of the coarse-graining mappings. Thus, both approaches employ symmetries to generate hierarchies but attain this via different procedures.

A key aspect of the approach presented here is the combination of objective and subjective elements, which opens the door to investigate not only properties of processes but also our inferences about them. 
Our approach to hierarchical belief updating is different from standard hierarchical Bayesian modelling~\cite{shiffrin2008survey}, which focuses on hierarchies in parameter space and generally does not allow to update coarse descriptions without involving the fine-grained layers. 
In addition, the efficiency of abstract beliefs as proposed here differs from the approach used in variational inference~\cite{blei2017variational}, which simplifies the class of models considered at the cost of introducing approximation errors. In contrast, the method presented here achieves exact inference by carefully restricting the target of prediction. 
Future work may combine these approaches, considering ongoing efforts combining variational methods and renormalisation~\cite{friston2024from}.

At its core, the results presented here support the view that the hierarchical structure of systems of interest is mirrored by an analogous hierarchical structure in our beliefs about them. In this, these results resemble old cybernetics ideas about how knowledge should resemble structural properties of its target~\cite{conant1970every,ashby2011variety}, but arrive at these conclusions using modern tools of information theory and Bayesian filtering. 
Recent work also revisits cybernetics from a similar perspective but using category theory~\cite{virgo2025good,baltieri2025bayesian}, exploring the relationship between external processes and epistemic beliefs in terms of functors and adjunctions. Future work may explore the potential of combining the presented results with such techniques, and also with approaches that blend renormalisation with principles of statistical inference~\cite{berman2023bayesian} or information theory~\cite{koch2018mutual,lenggenhager2020optimal}.

Overall, the presented results contribute towards a principled formalisation of the general notion of abstraction, which is a fundamental open problem in cognitive science~\cite{ho2019value,ho2022people,de2023goals} and artificial intelligence~\cite{chollet2019measure,geiger2025causal,konidaris2019necessity}. 
There have been important advances formalising abstractions in the context of reinforcement learning~\cite{abel2020thesis,allen2023thesis}, including approaches leveraging symmetries~\cite{van2020mdp}, but most efforts have focused on scenarios where the environment is fully observable.
While partially observable scenarios are formally solvable by turning them into fully observable problems by computing their Bayesian beliefs~\cite{kaelbling1998planning}, this approach doesn't work in practice because beliefs often are too high-dimensional. 
The approach presented here may be used to coarse-grain Bayesian beliefs in a principled manner, which could be used to address problems that are intractable to current methods. 
Moreover, there is potential for exploring these ideas leveraging architectures inspired by geometric deep learning~\cite{bronstein2021geometric} focused on equivariance --- see e.g. Refs.~\cite{di2025shaping,keller2025flow}.

A key challenge to make emergence useful for learning scenarios is how to design efficient methods to identify emerging coarse-graining levels. In addition to existent methods based on linear methods~\cite{barnett2023dynamical} or differentiable estimators~\cite{mcsharry2024learning}, the results presented here suggest that symmetries could provide yet another route to do this. 
Moreover, focusing on more general forms of weak symmetry could open the door to efficient graph-colouring algorithms while expanding the scope of applicability~\cite{makse2025symmetries}, constituting a promising direction for future work.

\section*{Acknowledgements}

Some of the ideas explored here were motivated by inspiring discussions with Adam Goldstein and Hernan Makse during the Softmax UK gathering 2025. 
I am also grateful for feedback and insightful conversations with Manuel Baltieri, Alexander Boyd, Alexander Gietelink Oldenziel, Ryota Kanai, Franz Nowak, Lucas Teixeira, and Nathaniel Virgo. 
This work has been supported by the UK ARIA Safeguarded AI programme, the PIBBSS Affiliateship programme, and Open Philanthropy.

\bibliography{references}
\appendix

\section{Proof of \autoref{teo:HMM}}
\label{app:proof_teo_HMM}

\begin{proof}
    Using \autoref{eq:sym_HMMa} one can show that
    \begin{align}
        p\big(z_{t+1}|g\gdot x_t\big) 
        &=
        \sum_{x_{t+1}\in\phi_G^{-1}(z_{t+1})}
        p\big(x_{t+1}|g\gdot x_t\big) \\
        &=
        \sum_{x_{t+1}\in\phi_G^{-1}(z_{t+1})}
        p\big(x_{t+1}|x_t\big) \\
        &=
        p\big(z_{t+1}|x_t\big).
    \end{align}
    A similar calculation using \autoref{eq:sym_HMMb} leads to 
    \begin{equation}
        p\big(v_{t}|g\gdot x_t\big) = p\big(v_{t}|x_t\big). 
    \end{equation}
    These results imply that
        $p\big(z_{t+1}|x_t\big) 
        =
        p\big(z_{t+1}|z_t\big)$ and 
        $p\big(v_t|x_t\big) 
        =
        p\big(v_t|z_t\big)$, 
    which, combined with \autoref{eq:initial_sym}, imply that
    \begin{align}
    p(&\finfut{z}{0}{t+1}\!,\finfut{v}{0}{t})
    =
    \!\!\!\!\!\!
    \sum_{
    \substack{\finfut{x}{0}{t+1}\in \phi_G^{-1}(\finfut{z}{0}{t+1})\\
    \finfut{y}{0}{t}\in \psi_G^{-1}(\finfut{v}{0}{t}) }
    }
    \!\!\!\!\!\!\!\!
    p(x_0)
    \!
    \prod_{\tau=0}^{t} 
    p( x_{\tau+1}| x_{\tau})
    p( y_{\tau}|x_{\tau} )  \nonumber\\
    =&
    \!\!\!\!\!
    \sum_{
    \substack{\finfut{x}{0}{t}\in \phi_G^{-1}\!(\finfut{z}{0}{t})\\
    \finfut{y}{0}{t-1}\in \psi_G^{-1}\!(\finfut{v}{0}{t-1}\!) }
    }
    \!\!\!\!\!\!\!\!\!\!
    p( z_{t+1}| x_{t}) p( v_{t}|x_{t} ) 
    p(x_0)
    \!
    \prod_{\tau=0}^{t-1} 
    \!
    p( x_{\tau+1}| x_{\tau})
    p( y_{\tau}|x_{\tau} )  \nonumber\\
    =&
    \,p( z_{t+1}| z_{t}) p( v_{t}|z_{t} )
    \!\!\!\!\!\!\!\!\!\!\!\!
    \sum_{
    \substack{\finfut{x}{0}{t}\in \phi_G^{-1}(\finfut{z}{0}{t})\\
    \finfut{y}{0}{t-1}\in \psi_G^{-1}(\finfut{y}{0}{t-1}) }
    }
    \!\!\!\!\!\!\!\!\!\!\!
    p(x_0)
    \prod_{\tau=0}^{t-1} 
    p( x_{\tau+1}| x_{\tau})
    p( y_{\tau}|x_{\tau} ) \nonumber \\
    =&\ldots \nonumber\\
    =&
    \,p(z_0)
    \prod_{\tau=0}^{t} 
    p( z_{\tau+1}| z_{\tau})
    p( v_{\tau}| z_{\tau} ). 
    \end{align}
    
\end{proof}

\section{Details of Hopfield network}
\label{app:hopfield}

For the numerical experiments, Hopfield networks of $20\times 20$ neurons were trained on the rotations and reflections of the letters $G$, $J$, and $R$, giving rise to $m=3\times 8 = 24$ patterns. To enhance the capacity of the network to store non-orthogonal patterns, these patterns were encoded into synaptic weights using a variation of \autoref{eq:hebb} known as the \textit{pseudo-inverse learning rule}~\cite{kanter1987associative}. This approach yields a synaptic weight matrix $W$ given by
\begin{equation}
    W = \frac{1}{m} Q \Big( \frac{1}{m} Q^\top Q \Big)^{-1} Q^\top,
\end{equation}
where $Q=[q_1,\dots,q_m]$ is an $n\times m$ matrix holding $m<n$ target patterns $q_\mu=(q_\mu^1,\ldots,q_\mu^n) \in\{-1,1\}^n$ as columns. 

For this synaptic encoding, each pattern $q_\mu$ is an eigenvector of $W$ with eigenvalue $1$, while vectors orthogonal to the span of the patterns are also eigenvectors but with zero eigenvalue. This implies that
\begin{equation}\label{eq:prod}
    W x_t 
    = 
    \sum_{\mu=1}^m 
    \big(
    h_\mu\gdot x_t
    \big)
    q_\mu,
\end{equation}
where $(a\cdot b) = \sum_i a_ib_i/n$ is the normalised dot product, and $h_\mu$ are the rows of $(Q^\top Q)^{-1}Q^\top$, which satisfy $(h_\mu\cdot~q_\nu) = \delta_\mu^\nu$. Note that if $q_1,\ldots,q_m$ are orthonormal, then $h_\mu=q_\mu$. 
Thus, in the pseudo-inverse setting, one can define a generalised Mattis magnetisation as $Z_t = (Z_t^1,\dots,Z_t^m)$ with
\begin{equation}
    Z_t^\mu := (h_\mu\gdot X_t)/n\in[-1,1], 
\end{equation}
which reduces to the original when the patterns are orthonormal. Using an analogous reasoning as in \autoref{sec:hopfield}, one can show this magnetisation can be retrieved via orbits of permutations satisfying
\begin{equation}
\big( h_\mu\gdot (\rho\gdot x) \big)
= 
( h_\mu\gdot  x )
\quad
\forall \mu=1,\ldots,m,
\forall x\in\{-1,1\}^n.
\end{equation}
It can be shown that the network dynamics are equivariant with respect to this symmetry, which makes this magnetisation informationally closed.

\end{document}